\documentclass[letterpaper, 10 pt, conference]{ieeeconf}  

\IEEEoverridecommandlockouts                              

\usepackage{graphics} 
\usepackage{epsfig} 
\usepackage{mathptmx} 
\usepackage{times} 
\usepackage{amsmath} 
\usepackage{amssymb}  
\usepackage{cite}               
\usepackage{color}
\usepackage{float}
\usepackage{listings}             
\newtheorem{definition}{Definition}

\newtheorem{theorem}{Theorem}
\newtheorem{num-example}{Numerical example}

\makeatletter
\newenvironment{sqcases}{%
	\matrix@check\sqcases\env@sqcases
}{%
	\endarray\right.%
}
\def\env@sqcases{%
	\let\@ifnextchar\new@ifnextchar
	\left\lbrack
	\def\arraystretch{1.2}%
	\array{@{}l@{\quad}l@{}}%
}
\makeatother

\title{\LARGE \bf
	The conservative lock-in range for PLL with lead-lag filter \\and triangular phase detector characteristic
}

\author{Blagov~M.V.$^{a}$
	Kuznetsov~N.V.$^{b}$, 
	Lobachev~M.Y.$^{c}$, 
	Yuldashev~M.V.$^{d}$, 
	Yuldashev~R.V.$^{d}$  
	\thanks{$^{a}$Mikhail V. Blagov is with the Faculty of Mathematics and Mechanics,
		Saint Petersburg State University, Russia, 
		with the Faculty of Mathematical Information Technology,
		University of Jyv\"{a}skyl\"{a}, Finland}
	\thanks{$^{b}$Nikolay V. Kuznetsov is with the Faculty of Mathematics and Mechanics,
		Saint Petersburg State University, Russia, 
		with the Faculty of Mathematical Information Technology,
		University of Jyv\"{a}skyl\"{a}, Finland,
		with the Institute for Problems in Mechanical Engineering RAS, Russia
		{\tt\small nkuznetsov239@gmail.com}}%
	\thanks{$^{c}$Mikhail Y. Lobachev is with the Faculty of Mathematics and Mechanics,
		Saint Petersburg State University, Russia, 
		with the Industrial Management Department, LUT University, Finland}
	\thanks{$^{d}$Marat V. Yuldashev, Renat V. Yuldashev are with the Faculty of Mathematics and Mechanics,
	Saint Petersburg State University, Russia}
}

\begin{document}

	\maketitle
	\thispagestyle{empty}
	\pagestyle{empty}

	\begin{abstract}
		
	In the present work, a second-order PLL with lead-lag loop filter and triangular phase detector characteristic is analysed. 
	An exact value of the conservative lock-in range is obtained for the considered model.
	The solution is based on analytical integration of the considered model on the linear segments.

	\end{abstract}

	\section{INTRODUCTION}
	\label{sec:introduction}

The interest to study phase-locked loops (PLL) comes from their wide applications. 
Initially described by A.~Appleton in 1923 \cite{Appleton-1923} and H.~Bellescize \cite{Bellescize-1932}, these circuits became widely spread in wireless communications \cite{DuS-2010-communication, Rouphael-2014, BestKLYY-2016, Cho-2006, Ho-2005-Communication, Helaluddin-2008, rosenkranz2016receiver}, GPS navigation \cite{KaplanH-2017-GpS}, gyroscope systems \cite{AaltonenH-2010, KuznetsovKBTYY-2022-GN}, computer architectures \cite{Kolumban-2005, Best-2018}, and others. 

First ideas of mathematical analysis of such systems belong to Italian academician F.~Tricomi \cite{Tricomi-1933} and are based on the analysis of system phase portraits. 
These ideas were further developed in works of A.A.~Andronov \cite{AndronovVKh-1937}.
Fundamental monographs devoted to the problems of numerical simulation and analysis of PLL were published in 1966 by F.~Gardner \cite{Gardner-1966}, A.~Viterbi \cite{Viterbi-1966}, V.V.~Shakhgildyan, and A.A.~Lyakhovkin \cite{ShakhgildyanL-1966}. 
These books are devoted mostly to engineering approaches of two-dimensional PLL models analysis.  

In this article, we consider a PLL with lead-lag loop filter and triangular phase detector characteristic.
Nonlinear analysis of this model and estimates of the global stability domain were conducted in \cite{Kapranov-1956, Gubar-1961, Shakhtarin-1969, Lindsey-1972, EndoT-1986, Stensby-2011}.
Basing on these works, we analytically obtain an exact formula for the conservative lock-in range for the first time.
This characterisctic considers the ability of PLL to synchronize in a short time and related to the Gardner problem \cite{Gardner-2005-book, LeonovKYY-2015-TCAS}.

\section{Mathematical model and hold-in range}
\begin{figure}[h]
	\centering
	\includegraphics[width=\linewidth]{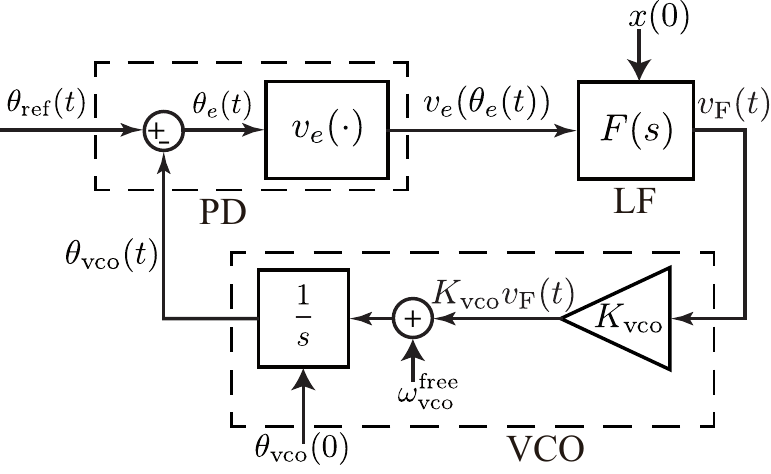}
	\caption{Baseband model of analog PLLs.}
	\label{fig:PLL-model}
\end{figure}

Consider analog PLL baseband model in Fig.~\ref{fig:PLL-model} \cite{Gardner-2005-book, Viterbi-1966, Best-2007, LeonovKYY-2012-TCASII, LeonovKYY-2015-SIGPRO}. 
Here $\theta_{\rm ref}(t) = \omega_{\rm ref}t + \theta_{\rm ref}(0)$ is a phase of the reference signal,
a phase of the VCO is $\theta_{\rm vco}(t)$, 
$\theta_e(t) = \theta_{\rm ref}(t) - \theta_{\rm vco}(t)$ is a phase error.
A phase detector (PD) generates a signal $v_e(\theta_e(t))$
where $v_e(\cdot)$ is a characteristic of the phase detector. 
In the present paper, a piecewise-linear PD characteristic, which is continuous and corresponds to square waveforms of the reference and the VCO signals, is considered:
\begin{equation}\label{eq:piecewise-linear PD characteristic}
\begin{aligned}
&v_e(\theta_e)
=
\begin{cases}
& \frac{2}{\pi}\theta_e - 4 m, 
\qquad - \frac{\pi}{2} + 2 \pi m \leq \theta_e(t) < \frac{\pi}{2} + 2 \pi m,\\
& - \frac{2}{\pi}\theta_e + 2 + 4m,
\qquad \frac{\pi}{2} + 2 \pi m \leq \theta_e(t) < - \frac{\pi}{2} + 2 \pi (m + 1),
\end{cases}
\end{aligned}
\end{equation}
here $m \in \mathbb{Z}$ (see Fig.~\ref{fig:triangular PD characteristic}).

\begin{figure}[h]
	\centering
	\includegraphics[width=0.9\linewidth]{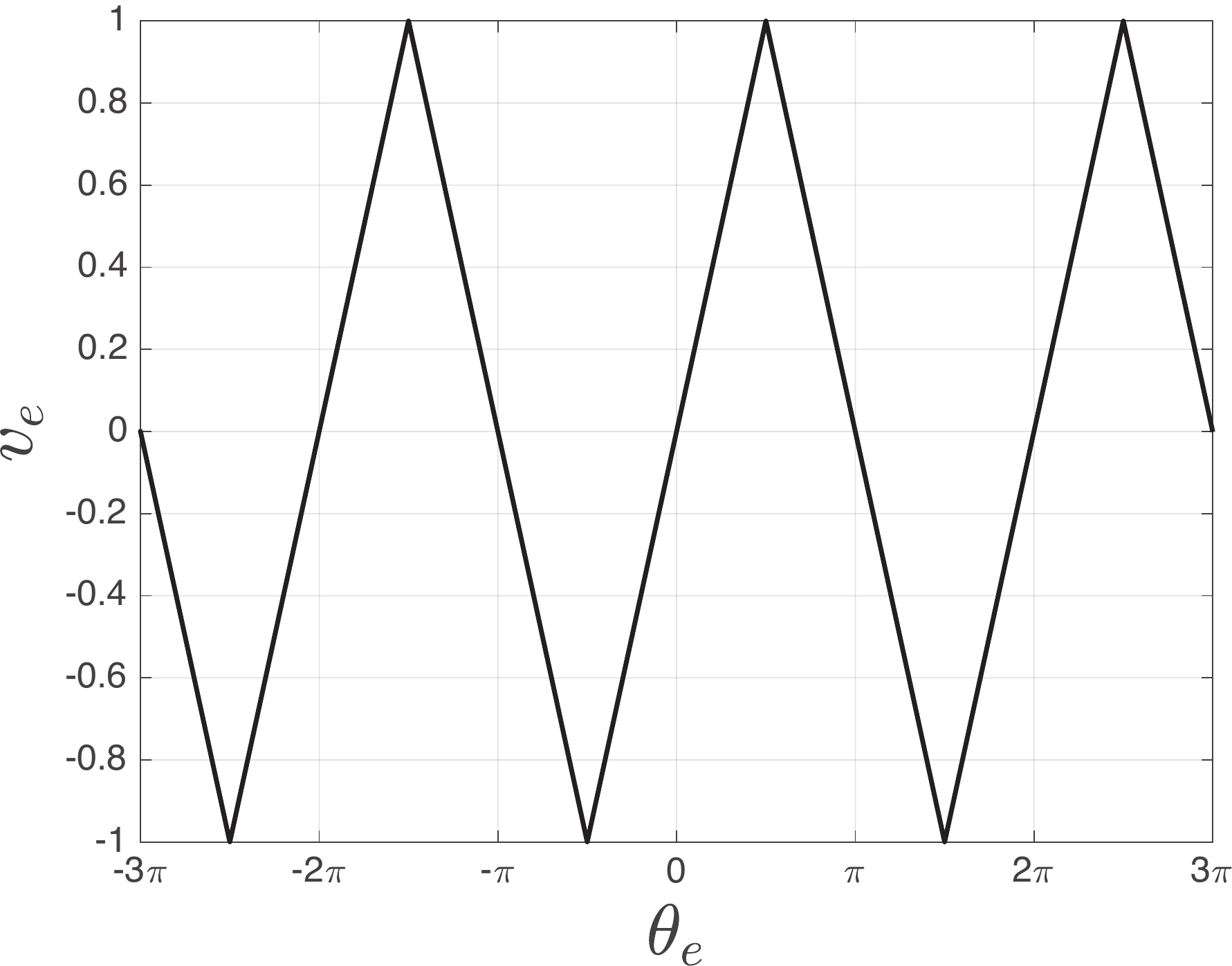}
	\caption{Triangular PD characteristic.}
	\label{fig:triangular PD characteristic}
\end{figure}

The state of the loop filter is represented by 
$x(t)\in\mathbb{R}$
and the transfer function is\footnote{If $\tau_2=0$ then such filter is called a lag filter, if $\tau_2>0$ then it is called a lead-lag filter \cite{Gardner-2005-book}.}
\begin{equation*}
\begin{aligned}
& F(s)=\frac{1+\tau_2s}{1+(\tau_1+\tau_2)s},\quad
\tau_1>0,\ \tau_2\ge0.
\end{aligned}
\end{equation*}
The output of the loop filter $v_{\rm F}(t) = \frac{1}{\tau_1+\tau_2}x + \frac{\tau_2}{\tau_1 + \tau_2}v_e(\theta_e)$ is used to control the VCO frequency $\omega_{\rm vco}(t)$, which is proportional to the control voltage:
\begin{equation*}
\begin{aligned}
& \omega_{\rm vco}(t) = \dot\theta_{\rm vco}(t) = \omega_{\rm vco}^{\rm free} + K_{\rm vco} v_{\rm F}(t)
\end{aligned}
\end{equation*}
where $K_{\rm vco}>0$ is a gain and $\omega_{\rm vco}^{\rm free}$ is a free-running frequency of the VCO. 

The behavior of PLL baseband model in the state space is described by a second-order nonlinear ODE:
\begin{equation}\label{eq:PLL-model}
\begin{aligned}
& \dot x = -\frac{1}{\tau_1 + \tau_2}x + \frac{\tau_1}{\tau_1 + \tau_2}v_e(\theta_e),\\
& \dot\theta_e = \omega_e^{\rm free} - K_{\rm vco}\Big(\frac{1}{\tau_1+\tau_2}x + \frac{\tau_2}{\tau_1 + \tau_2}v_e(\theta_e)\Big).
\end{aligned}
\end{equation}
where $\omega_e^{\rm free} = \omega_{\rm ref} - \omega_{\rm vco}^{\rm free}$ is a frequency error and $v_e(\theta_e)$ is defined in \eqref{eq:piecewise-linear PD characteristic}.
It is usually supposed that the reference frequency (hence, $\omega_e^{\rm free}$ too) can be abruptly changed and that the synchronization occurs between those changes.
Thus, existence of locked states, acquisition and transient processes after the reference frequency change are of interest.

The PLL baseband model in Fig.~\ref{fig:PLL-model} is locked if the phase error $\theta_e(t)$ is constant. 
For the locked states of practically used PLLs, the loop filter state is constant
too and, thus, the locked states of model in Fig.~\ref{fig:PLL-model} correspond to the equilibria of model \eqref{eq:PLL-model} \cite{KuznetsovLYY-2015-IFAC-Ranges}.

\begin{definition}\cite{KuznetsovLYY-2015-IFAC-Ranges, LeonovKYY-2015-TCAS, BestKLYY-2016}
	A \emph{hold-in range} is the largest symmetric interval of frequency errors $|\omega_e^{\rm free}|$ such that an asymptotically stable equilibrium exists and varies continuously while $\omega_e^{\rm free}$ varies continuously within the interval.
	\label{def:hold-in}
\end{definition}

Observe that system \eqref{eq:PLL-model} is $2\pi$-periodic in $\theta_e$ and has an infinite number of equilibria $(x^{\rm eq},\ \theta_e^{\rm eq})$ which satisfy
\begin{equation*}\label{eq:equations for equilibria}
\begin{aligned}
&v_e(\theta_e^{\rm eq}) = \frac{\omega_e^{\rm free}}{K_{\rm vco}},\\
&x^{\rm eq} = \frac{\tau_1\omega_e^{\rm free}}{K_{\rm vco}}.
\end{aligned}
\end{equation*}
From the boundedness of the PD characteristic it follows that there are no equilibria for sufficiently large $\omega_e^{\rm free}$.
Further we suppose that $\omega_e^{\rm free} < K_{\rm vco}$ and the equilibria are
\begin{equation}\label{eq:equilibria}
\begin{aligned}
&\left(\frac{\tau_1\omega_e^{\rm free}}{K_{\rm vco}},\ (-1)^m\frac{\frac{\pi}{2}\omega_e^{\rm free}}{K_{\rm vco}} + \pi m\right),\ m\in\mathbb{Z}.
\end{aligned}
\end{equation}
The characteristic polynomial of system \eqref{eq:PLL-model} linearized at stationary states \eqref{eq:equilibria} is
\begin{equation*}
\begin{aligned}
& \chi(\lambda) = \lambda^2 
+ 
\Big(\frac{1}{\tau_1 + \tau_2} + \frac{K_{\rm vco}\tau_2}{\tau_1 + \tau_2}v_e^\prime(\theta_e^{\rm eq})\Big)\lambda  
+ 
\frac{K_{\rm vco}}{\tau_1 + \tau_2}v_e^\prime(\theta_e^{\rm eq}).
\end{aligned}
\end{equation*}
The nonlinearity $v_e(\theta_e)$ decreases $\left(v_e^\prime(\pi - \frac{\frac{\pi}{2}\omega_e^{\rm free}}{K_{\rm vco}} + 2\pi m) = -\frac{2}{\pi} < 0\right)$ for $\frac{\pi}{2} + 2 \pi m \leq \theta_e(t) < - \frac{\pi}{2} + 2 \pi (m + 1)$, and equilibria 
\begin{equation*}
\begin{aligned}
& \left(\frac{\tau_1\omega_e^{\rm free}}{K_{\rm vco}},\ \pi - \frac{\frac{\pi}{2}\omega_e^{\rm free}}{K_{\rm vco}} + 2\pi m\right)
\end{aligned}
\end{equation*}
are saddles.
The nonlinearity $v_e(\theta_e)$ increases $\left(v_e^\prime(\frac{\frac{\pi}{2}\omega_e^{\rm free}}{K_{\rm vco}} + \pi m) = \frac{2}{\pi} > 0\right)$ for $- \frac{\pi}{2} + 2 \pi m \leq \theta_e(t) < \frac{\pi}{2} + 2 \pi m$, and equilibria
\begin{equation*}
\begin{aligned}
& \left(\frac{\tau_1\omega_e^{\rm free}}{K_{\rm vco}},\ \frac{\frac{\pi}{2}\omega_e^{\rm free}}{K_{\rm vco}} + 2\pi m\right)
\end{aligned}
\end{equation*}
are asymptotically stable ones, which can be either nodes, degenerate nodes or foci (see Appendix).
Since an asymptotically stable equilibrium exists for any frequency error $\omega_e^{\rm free} < K_{\rm vco}$, the hold-in range of model \eqref{eq:PLL-model} is $[0, \omega_h) = [0, K_{\rm vco})$ for any $\tau_1 > 0,\ \tau_2 \ge 0$.

\section{Global stability analysis}
\begin{definition}\cite{KuznetsovLYY-2015-IFAC-Ranges,LeonovKYY-2015-TCAS,BestKLYY-2016}
	A \emph{pull-in range} is the largest symmetric interval 
	of frequency errors $|\omega_e^{\rm free}|$ from the hold-in range such that an equilibrium is acquired for an arbitrary initial state.	
\end{definition}

\subsection{Pull-in range estimate by Lyapunov function}
To obtain an estimate for the pull-in range of system \eqref{eq:PLL-model}, 
we apply the direct Lyapunov method and the corresponding theorem on global stability for the cylindrical phase space
\begin{theorem}(see, e.g., \cite{LeonovK-2014-book, KuznetsovLYYKKRA-2020-ECC}).
	\label{th:Lyapunov-type lemma}
	If there is a continuous function $V(x, \theta_e): \mathbb{R}^{2}\to\mathbb{R}$ such that
	
	(i) $V(x, \theta_e + 2\pi) = V(x, \theta_e) \quad\forall x\in\mathbb{R}, \forall \theta_e\in\mathbb{R}$;
	
	(ii) for any solution $(x(t), \theta_e(t))$ of system \eqref{eq:PLL-model} the function $V(x(t), \theta_e(t))$ is nonincreasing;
	
	(iii) if $V(x(t), \theta_e(t))\equiv V(x(0), \theta_e(0))$, then $(x(t), \theta_e(t))\equiv~(x(0), \theta_e(0))$;
	
	(iv) $V(x, \theta_e)+\theta_e^2\to+\infty$ as $||x||+|\theta_e|\to+\infty$
	
	\noindent then  
	any trajectory of system \eqref{eq:PLL-model} tends to an equilibrium.
\end{theorem}\medskip

Following \cite{BakaevG-1965, LeonovK-2014-book}, consider the following Lyapunov function:
\begin{equation}\label{eq:Lyapunov function lead-lag}
\begin{aligned}
& V(x,\ \theta_e) = \frac{1}{2} (x - \frac{\tau_1\omega_e^{\rm{free}}}{K_{\rm vco}})^2
+\\
&+
\frac{\tau_1}{K_{\rm vco}}
\int\limits_0^{\theta_e} 
\Big(
v_e(\sigma) - \frac{\omega_e^{\rm free}}{K_{\rm vco}} + \beta_0|v_e(\sigma) - \frac{\omega_e^{\rm free}}{K_{\rm vco}}|
\Big)d\sigma
\end{aligned}
\end{equation}
where
\begin{equation*}
\begin{aligned}
& 
\beta_0=-
\frac{\int_0^{2\pi} (v_e(\sigma) - \frac{\omega_e^{\rm free}}{K_{\rm vco}})\ d\sigma}
{\int_0^{2\pi} |v_e(\sigma) - \frac{\omega_e^{\rm free}}{K_{\rm vco}}|\ d\sigma} > 0.
\end{aligned}
\end{equation*}
Such form of the integrand expression makes the Lyapunov function $2\pi$-periodic.
For triangular PD characteristic coefficient $\beta_0$ is
\begin{equation}\label{eq:beta_0}
\begin{aligned}
& \beta_0 = \frac{2\omega_e^{\rm{free}}K_{\rm vco}}{(\omega_e^{\rm{free}})^2+K_{\rm vco}^2}.
\end{aligned}
\end{equation}
The Lyapunov function derivative along the trajectories of system \eqref{eq:PLL-model} is
\begin{equation*}
\begin{aligned}
&\dot V(x, \theta_e)
=
-\frac{1}{\tau_1+\tau_2}\Big(
(x-\frac{\tau_1\omega_e^{\rm{free}}}{K_{\rm vco}})^2 
- \\
&-
\beta_0\tau_1(x-\frac{\tau_1\omega_e^{\rm{free}}}{K_{\rm vco}})(v_e(\theta_e) - \frac{\omega_e^{\rm free}}{K_{\rm vco}}) 
+\\
&+
\tau_1\tau_2(1 - \beta_0)(v_e(\theta_e) - \frac{\omega_e^{\rm free}}{K_{\rm vco}})^2
\Big).
\end{aligned}
\end{equation*}
If the loop filter parameters satisfy the inequality
\begin{equation}\label{eq:Sylvester criterion}
\begin{aligned}
& \beta_0 < 2(-\frac{\tau_2}{\tau_1} + \frac{\sqrt{\tau_2(\tau_1+\tau_2)}}{\tau_1})
\end{aligned}
\end{equation}
then the Lyapunov function derivative along the trajectories of system \eqref{eq:PLL-model} is as follows: 
\begin{equation*}
\begin{aligned}
&
\dot V(x,\theta_e)
< 0,\quad x \neq \frac{\tau_1\omega_e^{\rm{free}}}{K_{\rm vco}},\ v_e(\theta_e) \ne \frac{\omega_e^{\rm free}}{K_{\rm vco}}.
\end{aligned}
\end{equation*}
Since the derivative along any solution other than equilibria is not identically zero, condition \eqref{eq:Sylvester criterion} provides the global stability of the system.
Taking into account \eqref{eq:beta_0} and \eqref{eq:Sylvester criterion}, the following estimate for the pull-in range is obtained:
\begin{equation}\label{eq:omega_p_Lyapunov}
\begin{aligned}
& 
\omega_p >
\Big(
\frac{\tau_1}{2\sqrt{\tau_2(\tau_1+\tau_2)} - 2\tau_2} 
-\\
&-
\sqrt{\frac{\tau_1^2}{(2\sqrt{\tau_2(\tau_1 + \tau_2)} - 2\tau_2)^2} - 1}
\Big)
K_{\rm vco}.
\end{aligned}
\end{equation}

\subsection{Analysis of cycles of first and second kind}
Firstly, let us analyse the dissipativity domain. 
Consider the following Lyapunov function:
\begin{equation*}
\begin{aligned}
& V(x,\ \theta_e) = \frac{1}{2}\tau_1x^2.
\end{aligned}
\end{equation*}
Its derivative along the trajectories of system \eqref{eq:PLL-model} is:
\begin{equation*}
\begin{aligned}
&\dot V(x, \theta_e)
=
-\frac{\tau_1}{\tau_1 + \tau_2} x
\Big(
x - \tau_1v_e(\theta_e)
\Big).
\end{aligned}
\end{equation*}
If $|x| > \tau_1v_e(\theta_e)$, then $\dot V(x, \theta_e) < 0$.
Hence, $\limsup\limits_{t\to+\infty} |x(t)| < \tau_1$ and an estimate for the dissipativity domain is $|x(t)| < \tau_1$.

Using change of variables $z = - \frac{K_{\rm vco}}{\tau_1 + \tau_2}(x - \frac{\tau_1\omega_e^{\rm free}}{K_{\rm vco}})$, system \eqref{eq:PLL-model} becomes system (4.3) from \cite{LeonovSel-2002} with 
$\alpha = \frac{1}{\tau_1 + \tau_2}$, 
$\beta = \frac{1}{\tau_1 + \tau_2}K_{\rm vco}$, 
$a = \frac{\tau_2}{\tau_1 + \tau_2}K_{\rm vco}$.
Applying Theorem~4.1 from \cite{LeonovSel-2002} we get that any trajectory of system \eqref{eq:PLL-model} which is bounded in $\mathbb{R}^2$ tends to an equilibrium, hence, there are no the cycles of the first kind.
If there is a homoclinic orbit in the system, then it envelops an asymptotically stable equilibrium and a cycle of the second kind exists in this case due the dissipativity \cite{Gubar-1961} (thus, a homoclinic orbit does not determine the global stability and the pull-in range).

Thus, depending on the system parameters $K_{\rm vco},\ \tau_1,\ \tau_2$ there are three possibilities of the global stability loss in system \eqref{eq:PLL-model}:
\begin{itemize}
	\item disappearance of equilibria (in this case $[0, \omega_p) = [0, \omega_h)$
	\item appearance of separatrix cycle 
	\item appearance of semi-stable cycle (cycle of the second kind)
\end{itemize}
Applying Theorem~4.2 from \cite{LeonovSel-2002} it can be shown that there are no either the separatrix cycles or the cycles of the second kind in domain $x > x^{\rm eq}$.
Since the system is piecewise-linear, its trajectories can be analytically integrated (see Appendix) and exact frequency error values for separatrix and semi-stable cycles (hence, the pull-in range) can be obtained (see, e.g., \cite{Gubar-1961, Shakhtarin-1969, BlagovKLYY-2015, BlagovKKLYY-2016-IFAC, KuznetsovBAYY-2019}).

\section{Conservative lock-in range}

Although a PLL model can be globally stable, the acquisition process can take long time.
To decrease the synchronization time, a lock-in range concept is frequently exploited \cite{Gardner-2005-book, Kolumban-2005, Best-2007}.

\begin{definition}\cite{KuznetsovLYY-2015-IFAC-Ranges, LeonovKYY-2015-TCAS, BestKLYY-2016}
	A \emph{lock-in range} is the largest interval of frequency errors $|\omega_e^{\rm free}|$ from the pull-in range such that the PLL model being in an equilibrium, after any abrupt change of $\omega_e^{\rm free}$ within the interval acquires an equilibrium without cycle slipping ($\sup\limits_{t>0} |\theta_e(0) - \theta_e(t)| < 2\pi$).	
\end{definition}

From a mathematical point of view, system \eqref{eq:PLL-model} can initially be in an unstable equilibrium (at one of the saddles) or can acquire it by a separatrix after a change of $\omega_e^{\rm free}$ (see~\cite{KuznetsovBAYY-2019, KuznetsovLYYK-2020-IFACWC}). 
Corresponding behavior is not observed in practice: system state is disturbed by noise and can't remain in unstable equilibrium.
Thus, two cycle-slipping-related characteristics of the system can be considered: 
\textit{the lock-in range} $|\omega_e^{\rm free}|\in [0, \omega_l)$ where the equilibria are considered to be stable and 
\textit{the conservative lock-in range} $|\omega_e^{\rm free}|\in[0, \omega_l^c) \subset [0, \omega_l)$ which takes into account the unstable behavior described above.
In this article, we analyse the conservative lock-in range $[0, \omega_l^c)$.

\begin{figure*}[h]
	\centering
	\begin{minipage}[h]{0.4\linewidth}
		\includegraphics[width=\linewidth]{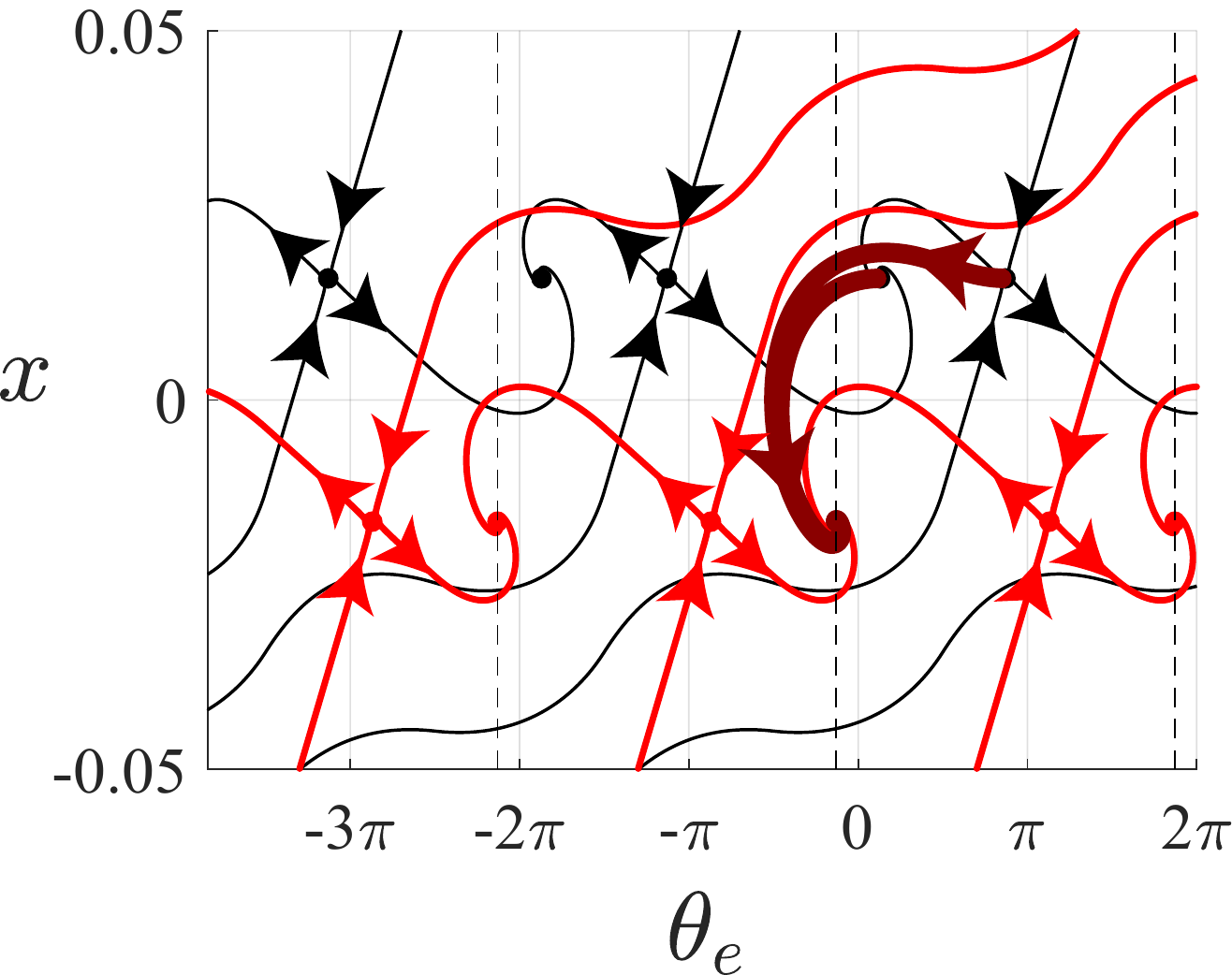}         	
	\end{minipage}
	$\qquad$
	\begin{minipage}[h]{0.4\linewidth}
		\includegraphics[width=\linewidth]{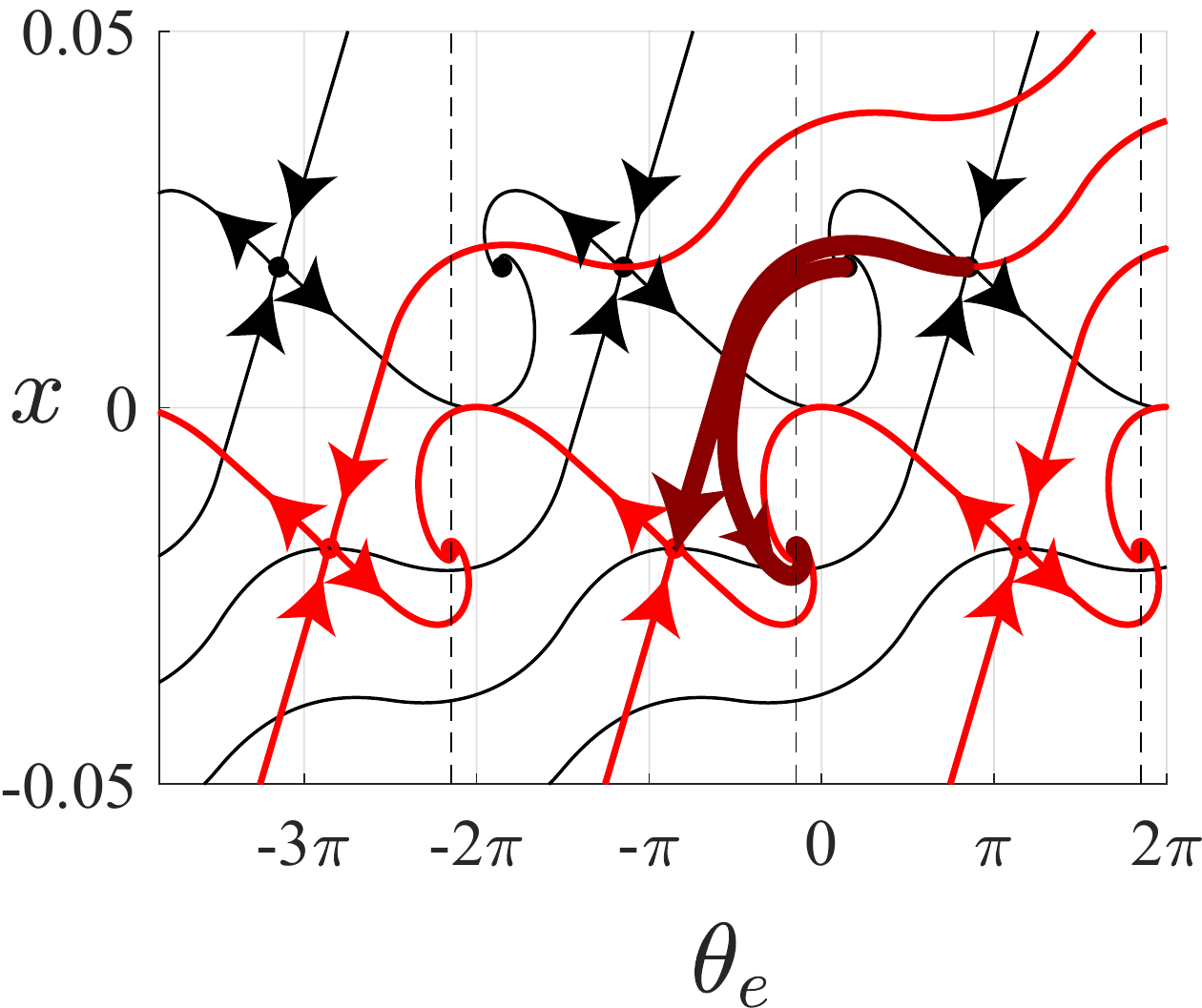}
	\end{minipage}
	
	\begin{minipage}[h]{0.4\linewidth}
		\includegraphics[width=\linewidth]{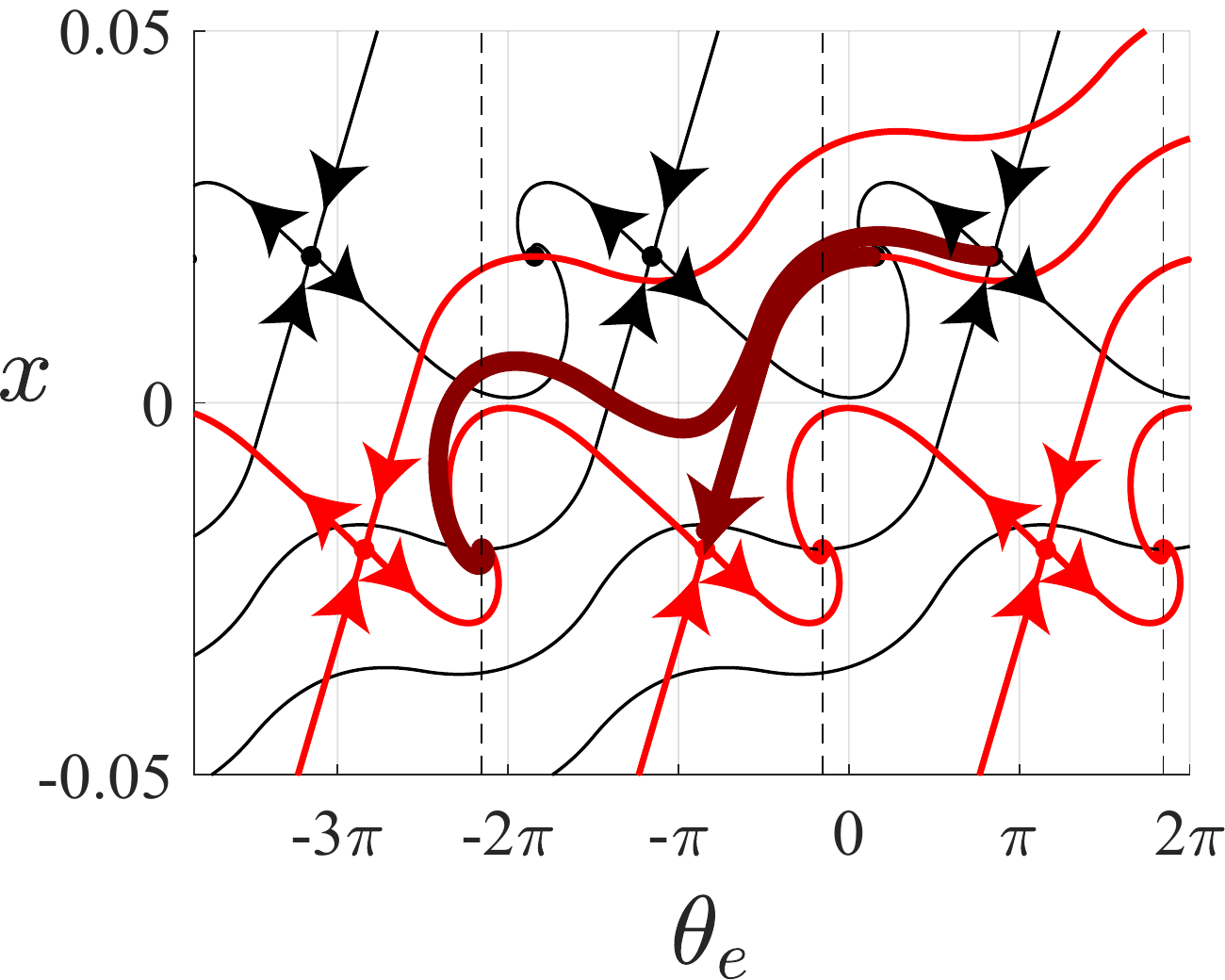}       	
	\end{minipage}
	$\qquad$
	\begin{minipage}[h]{0.4\linewidth}
		\includegraphics[width=\linewidth]{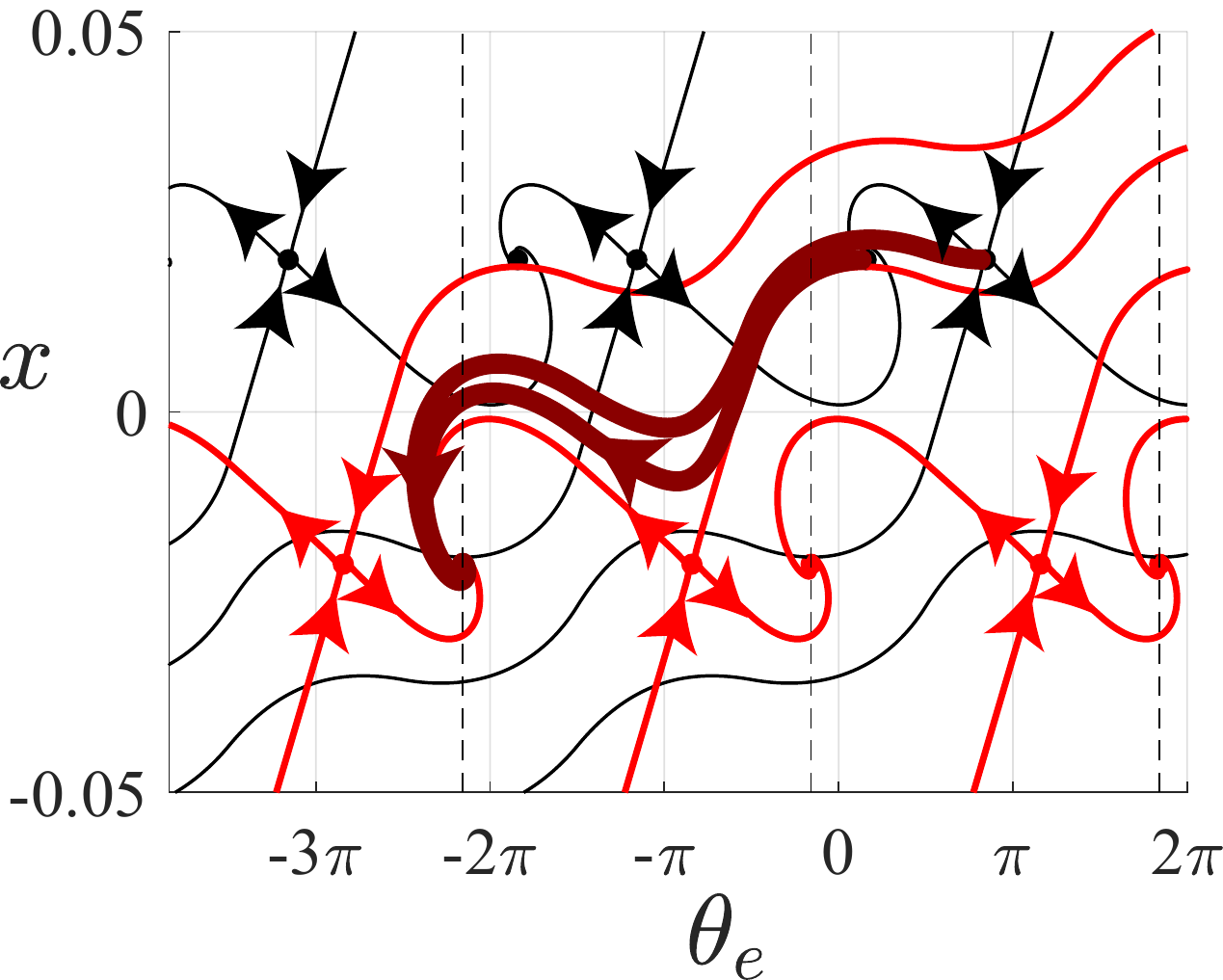}
	\end{minipage}
	\caption{Phase portraits for model \eqref{eq:PLL-model} with the following parameters:
		$F(s)=\frac{1+\tau_2s}{1+(\tau_1+\tau_2)s}$,
		$\tau_1 = 0.0633$,
		$\tau_2 = 0.0225$,
		$K_{\rm vco} = 250$.
		Black dots are equilibria of the model with positive $\omega_e^{\rm free}=|\omega|$.
		Red color is for the model with negative $\omega_e^{\rm free}=-|\omega|$.
		Separatrices pass in and out of the saddles equilibria.
		Upper left subfigure: $\omega = 65 < \omega_l^c$,
		upper right subfigure: $\omega = \omega_l^c \approx 73.732$,
		lower left subfigure: $\omega = \omega_l \approx 77.7583$,
		lower right subfigure: $\omega = 79 > \omega_l$.
	}
	\label{fig:lock-in illustration}
\end{figure*}

For the considered model boundary values $\omega_l$ and $\omega_l^c$ are determined as follows: The system being in an equilibrium state is exposed to an abrupt change of $\omega_e^{\rm free}$, and the corresponding trajectory of the system after the switch tends to the nearest unstable equilibrium by the corresponding saddle separatrix.
In other words, $\sup\limits_{t>0} |\theta_e(0) - \theta_e(t)| = \pi$ for $\theta_e(0) = 2\pi$ (see Fig.~\ref{fig:lock-in illustration}, lower left picture) and $\sup\limits_{t>0} |\theta_e(0) - \theta_e(t)| = 2\pi$ for $\theta_e(0) = 3\pi$ (see Fig.~\ref{fig:lock-in illustration}, upper right picture).
For a larger $\omega_e^{\rm free}$ supremum 
$\sup\limits_{t>0} |\theta_e(0) - \theta_e(t)| > 2\pi$ and cycle slipping occurs.
Since the lock-in range is defined as a half-open interval, boundary values $\omega_e^{\rm free} = \omega_l$ and $\omega_e^{\rm free} = \omega_l^c$ are not included in it.

Using changes of variables we represent system \eqref{eq:PLL-model} as the first-order differential equation \cite{Belyustina-1959, KuznetsovBAYY-2019}, analytically integrate it on the linear segments, formulate, and prove the theorem providing an exact value for the conservative lock-in range.

\begin{theorem}\label{theorem: concervative lock-in}
	The conservative lock-in frequency of model \eqref{eq:PLL-model} with triangular PD characteristic \eqref{eq:piecewise-linear PD characteristic} is $\omega_l^c$ which is the unique solution of system of two variables $(\omega_l^c,\ y_{\rm AB})$:
	\begin{equation}\label{eq:concervative lock-in}
	\begin{aligned}
	&\begin{cases}
	(2\omega_l^c)^2
	\Big(
	\sqrt{\frac{\tau_1 + \tau_2}{\frac{2}{\pi} K_{\rm vco}}} - \frac{\eta - \kappa}{\frac{2}{\pi}K_{\rm vco}}
	\Big)^{\frac{\kappa - \eta}{\kappa}}
	\Big(
	\sqrt{\frac{\tau_1 + \tau_2}{\frac{2}{\pi} K_{\rm vco}}} - \frac{\eta + \kappa}{\frac{2}{\pi}K_{\rm vco}}
	\Big)^{\frac{\kappa + \eta}{\kappa}}
	=\\
	=
	\Big(
	y_{\rm AB} - (\eta - \kappa)\frac{\omega_l^c + K_{\rm vco}}{\frac{2}{\pi}K_{\rm vco}}
	\Big)^{\frac{\kappa - \eta}{\kappa}}
	\Big(
	y_{\rm AB} -  (\eta + \kappa)\frac{\omega_l^c + K_{\rm vco}}{\frac{2}{\pi}K_{\rm vco}}
	\Big)^{\frac{\kappa + \eta}{\kappa}},\\
	\begin{sqcases}
	\big(y_{\rm AB} - (\xi - \rho)\frac{\omega_l^c + K_{\rm vco}}{\frac{2}{\pi} K_{\rm vco}}\big)^{\frac{\rho - \xi}{\rho}}
	\big(y_{\rm AB} - (\xi + \rho)\frac{\omega_l^c + K_{\rm vco}}{\frac{2}{\pi} K_{\rm vco}}\big)^{\frac{\rho + \xi}{\rho}}
	=\\
	=
	(\kappa - \eta + \xi - \rho)^{\frac{\rho - \xi}{\rho}}
	\cdot\\
	\cdot
	(\kappa - \eta + \xi + \rho)^{\frac{\rho + \xi}{\rho}}
	\big(\frac{K_{\rm vco} - \omega_l^c}{\frac{2}{\pi} K_{\rm vco}}
	\big)^2, \quad \text{if} \quad \xi > 1, \\
	\frac{K_{\rm vco} + \omega_l^c}
	{K_{\rm vco} + \omega_l^c - \frac{2}{\pi}K_{\rm vco}y_{\rm AB}} + \ln(2|y_{\rm AB} - \frac{\pi(K_{\rm vco} + \omega_l^c)}{2 K_{\rm vco}}|)
	=\\
	=
	\frac{1}{\kappa - \eta + 1} + \ln\left(2(\kappa - \eta + 1)\frac{\pi(K_{\rm vco} - \omega_l^c)}{2K_{\rm vco}}\right), \quad \text{if} \quad \xi = 1, \\
	\frac{1}{2}\ln(y_{\rm AB}^2 - 2\xi y_{\rm AB}\frac{\pi(K_{\rm vco} + \omega_l^c)}{2 K_{\rm vco}} + (\frac{\pi(K_{\rm vco} + \omega_l^c)}{2 K_{\rm vco}})^2)
	- \\
	-
	\frac{\xi}{\rho}
	\arctan\Big(
	\frac{y_{\rm AB} - \xi\frac{\pi(K_{\rm vco} + \omega_l^c)}{2 K_{\rm vco}}}
	{-(\frac{\pi(K_{\rm vco} + \omega_l^c)}{2 K_{\rm vco}})\rho}
	\Big)
	+
	\frac{\xi}{\rho}
	\arctan\Big(
	\frac{\kappa - \eta + \xi}
	{\rho}
	\Big)
	=\\
	=
	\frac{1}{2}\ln\left(
	\left((\kappa - \eta)^2 + 2\xi(\kappa - \eta) + 1\right)
	\left(\frac{K_{\rm vco} - \omega_l^c}{\frac{2}{\pi}K_{\rm vco}}\right)^2\right)
	+\\
	+
	\frac{\pi \xi}{\rho}, \quad \text{if} \quad \xi < 1 
	\end{sqcases}
	\end{cases}
	\end{aligned}
	\end{equation}
	where
	\begin{equation*}
	\begin{aligned}
	&\xi = \frac{\frac{2}{\pi}\tau_2 K_{\rm vco} + 1}{2\sqrt{\frac{2}{\pi}K_{\rm vco}(\tau_1+\tau_2)}},\quad
	\eta = \frac{\frac{2}{\pi}\tau_2K_{\rm vco} - 1}{2\sqrt{\frac{2}{\pi}K_{\rm vco}(\tau_1+\tau_2)}},\\
	&\rho = \sqrt{|\xi^2 - 1|},\quad 
	\kappa = \sqrt{\eta^2 + 1}.
	\end{aligned}
	\end{equation*}
\end{theorem}

\begin{proof}[Proof of Theorem~\ref{theorem: concervative lock-in}]
	The proof given in Appendix is based on the fact that system \eqref{eq:PLL-model} is piecewise-linear and can be integrated analytically on the linear segments.
\end{proof}

	\section{Computer simulation}
	Based on Theorem~\ref{theorem: concervative lock-in} an analytical-numerical method of the conservative lock-in range calculation was implemented (see Appendix~B and Fig.~\ref{fig:omega_l_c}).

	\begin{figure}[h]
	\centering
	\includegraphics[width=\linewidth]{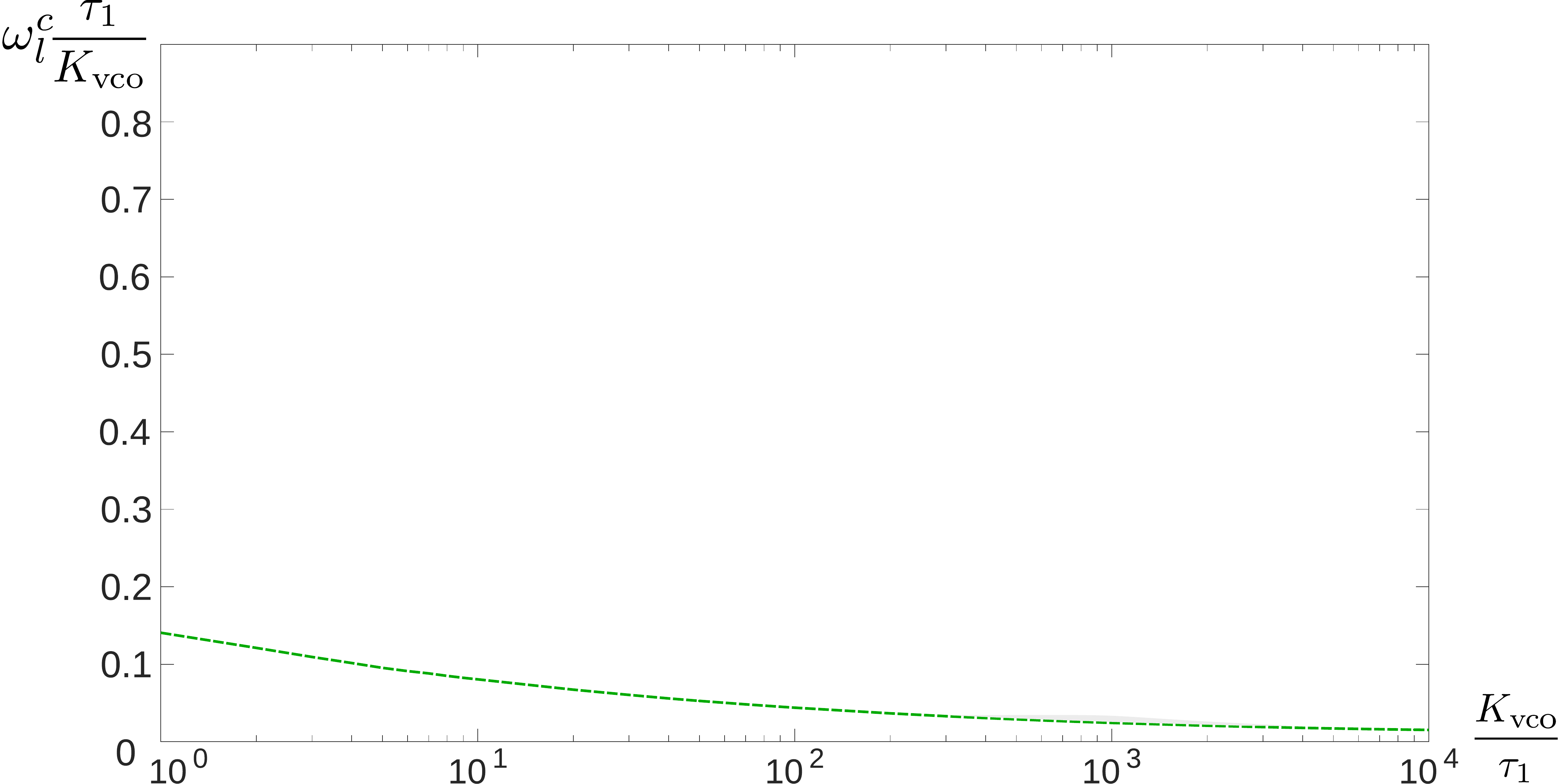}
	\caption{The conservative lock-in frequency.
		Parameters: 
		$\tau_1 = 0.5$,
		$\tau_2 = 0.0225$. 
	}
	\label{fig:omega_l_c}
\end{figure}
	
	\section{Conclusions}
	\label{sec:conclusion}
	
   In this work, the exact value of the conservative lock-in range was obtained for a classical PLL with lead-lag filter and triangular phase detector characteristic.

	\addtolength{\textheight}{-1cm}

	\section*{Appendix A: Proof of Theorem~\ref{theorem: concervative lock-in}}
	
	\begin{proof}
		Let's find the conservative lock-in range of model \eqref{eq:PLL-model} with triangular PD characteristic \eqref{eq:piecewise-linear PD characteristic}.
		The conservative lock-in frequency can be determined by such an abrupt change of $\omega_e^{\rm free}$ that the corresponding trajectory tends to the nearest unstable equilibrium (by the corresponding separatrix).
		Suppose that initially the frequency error was equal to $\omega_e^{\rm free} = -\omega < 0$, but then changed to $\omega_e^{\rm free} = \omega > 0$.
		Hence, initially the system was in equilibrium
		$x^{\rm eq} = -\frac{\tau_1\omega}{K_{\rm vco}},\quad
		\theta_e^{\rm eq} = -\pi + \frac{\frac{\pi}{2}\omega}{K_{\rm vco}}$, but after the switch the corresponding trajectory tends to
		$x^{\rm eq} = \frac{\tau_1\omega}{K_{\rm vco}},\quad
		\theta_e^{\rm eq} = \frac{\frac{\pi}{2}\omega}{K_{\rm vco}}$ without cycle slipping if $\omega < \omega_l^c$.
		
		Such $\omega_l^c$ is determined by such frequency error $\omega_e^{\rm free}$ that a trajectory being in unstable equilibrium (before the switch)
		$x^{\rm eq} = -\frac{\tau_1\omega_l^c}{K_{\rm vco}},\quad
		\theta_e^{\rm eq} = -\pi + \frac{\frac{\pi}{2}\omega_l^c}{K_{\rm vco}}$
		tends to the closest unstable equilibrium (after the switch)
		$x^{\rm eq} = \frac{\tau_1\omega_l^c}{K_{\rm vco}},\quad
		\theta_e^{\rm eq} = \pi - \frac{\frac{\pi}{2}\omega_l^c}{K_{\rm vco}}$
		by the corresponding separatrix.	
		Thus, the conservative lock-in frequency $\omega_l^c$ corresponds to the case
		\begin{equation}\label{eq:lock-in relations with Q}
		\begin{aligned}
		& -\frac{\tau_1\omega_l^c}{K_{\rm vco}}
		=
		Q\left(\frac{\frac{\pi}{2}\omega_l^c}{K_{\rm vco}} - \pi,\ \omega_l^c\right)
		\end{aligned}
		\end{equation}
		where $\frac{\tau_1\omega_e^{\rm free}}{K_{\rm vco}}$ is $x$-coordinate of equilibrium of model \eqref{eq:PLL-model} and $x = Q(\theta_e, \omega_e^{\rm free})$ is the lower separatrix of saddle equilibrium $(\frac{\tau_1\omega_e^{\rm free}}{K_{\rm vco}},\ \pi - \frac{\frac{\pi}{2}\omega_e^{\rm free}}{K_{\rm vco}})$ (see Fig.~\ref{fig:Q}).

		\begin{figure}[h]
			\centering
			\includegraphics[width=0.93\linewidth]{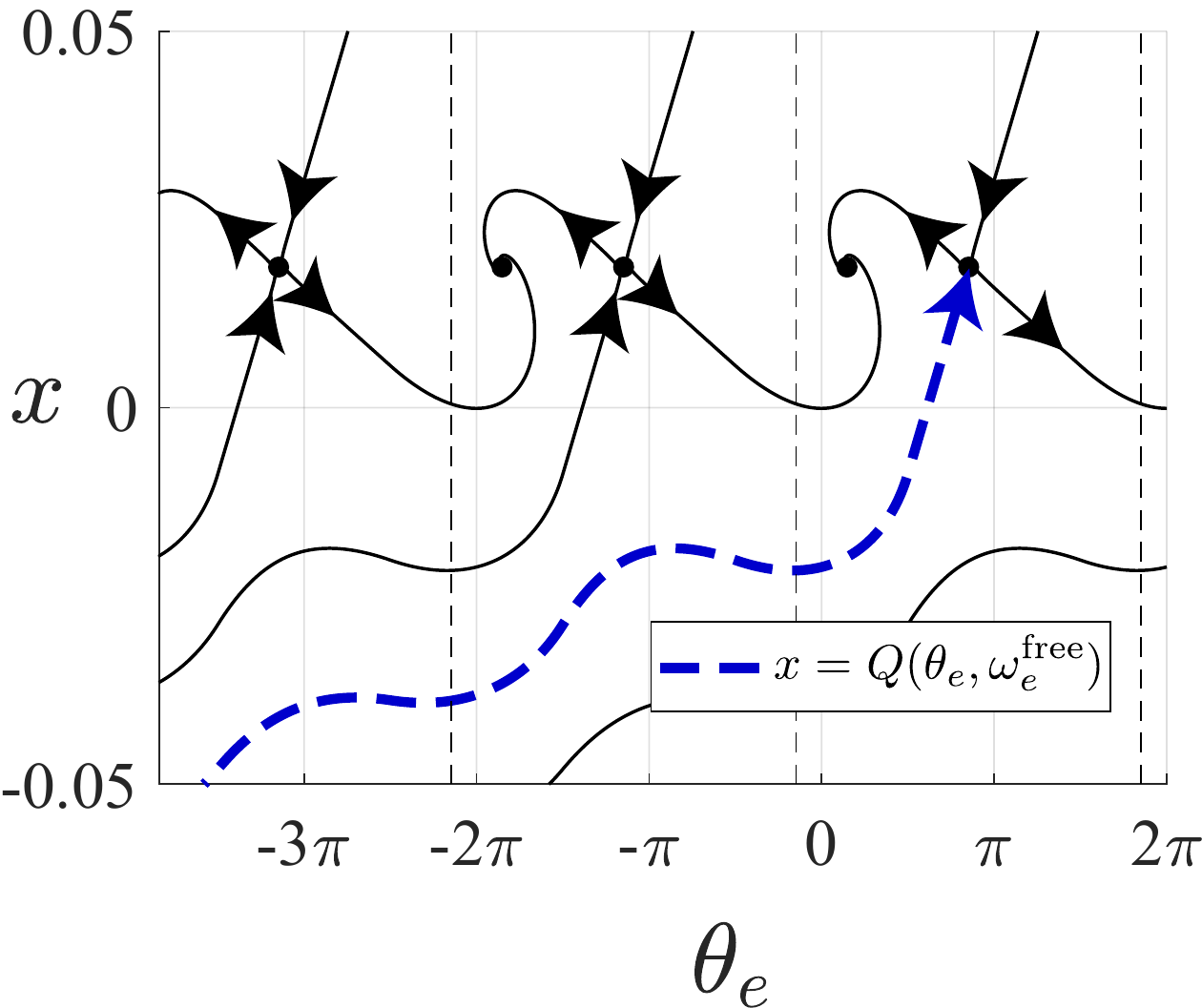}
			\caption{Separatrix $x = Q(\theta_e, \omega_e^{\rm free})$ of the phase plane of \eqref{eq:PLL-model}.
				Parameters:
				$\tau_1 = 0.0633$,
				$\tau_2 = 0.0225$,
				$K_{\rm vco} = 250$,
				$\omega_e^{\rm free} = 73.732$.
			}
			\label{fig:Q}
		\end{figure}
	
		After the change of variables
		\begin{equation}\label{eq:change of variables from x to y}
		\begin{aligned}
		& y = \sqrt{\frac{\pi(\tau_1 + \tau_2)}{2 K_{\rm vco}}}\omega_e^{\rm free}  - \sqrt{\frac{\pi K_{\rm vco}}{2(\tau_1 + \tau_2)}}
		\Big(x + \tau_2 v_e(\theta_e)\Big),\\
		&\tau = \sqrt{\frac{2 K_{\rm vco}}{\pi(\tau_1 + \tau_2)}} t
		\end{aligned}
		\end{equation}
		system \eqref{eq:PLL-model} in intervals $\theta_e(t)\in (- \frac{\pi}{2} + 2 \pi m,\ \frac{\pi}{2} + 2 \pi m)$ and $ \theta_e(t)\in(\frac{\pi}{2} + 2 \pi m,\ - \frac{\pi}{2} + 2 \pi (m + 1))$, $m\in\mathbb{Z}$ is represented as follows:
		\begin{equation}\label{eq:PLL-triangular-after_change_of_variables}
		\begin{aligned}
		& \dot y
		=
		- 
		\frac{\pi}{2}v_e(\theta_e)
		-
		\frac{\sqrt{\pi}}{\sqrt{2 K_{\rm vco}(\tau_1 + \tau_2)}}
		(1 + K_{\rm vco}\tau_2v_e^\prime(\theta_e)
		)y
		+
		\frac{\pi\omega_e^{\rm free}}{2K_{\rm vco}} ,\\
		&\dot \theta_e = y.
		\end{aligned}
		\end{equation}

		Upper separatrix $y = S(\theta_e)$ of the phase plane of \eqref{eq:PLL-triangular-after_change_of_variables} corresponds to separatrix $x = Q(\theta_e, \omega_e^{\rm free})$ from \eqref{eq:PLL-model} (see Fig.~\ref{fig:Q} and Fig.~\ref{fig:separatrix integration}) and has the form
		\begin{equation*}
		\begin{aligned}
		&S(\theta_e) = \sqrt{\frac{\pi(\tau_1 + \tau_2)}{2 K_{\rm vco}}}\omega_e^{\rm free} -\\
		&-
		\sqrt{\frac{\pi K_{\rm vco}}{2(\tau_1 + \tau_2)}}
		\Big(Q(\theta_e,\ \omega_e^{\rm free}) + \tau_2 v_e(\theta_e)\Big).
		\end{aligned}
		\end{equation*}
		Thus, relation \eqref{eq:lock-in relations with Q} takes the form
		\begin{equation}\label{eq:S and omega_l_c}
		\begin{aligned}
		&S\left(\frac{\pi\omega_l^c}{2K_{\rm vco}} - \pi\right) 
		= 
		2\omega_l^c\sqrt{\frac{\pi(\tau_1 + \tau_2)}{2 K_{\rm vco}}}.
		\end{aligned}
		\end{equation}

\begin{figure}[h]
	\centering
	\includegraphics[width=\linewidth]{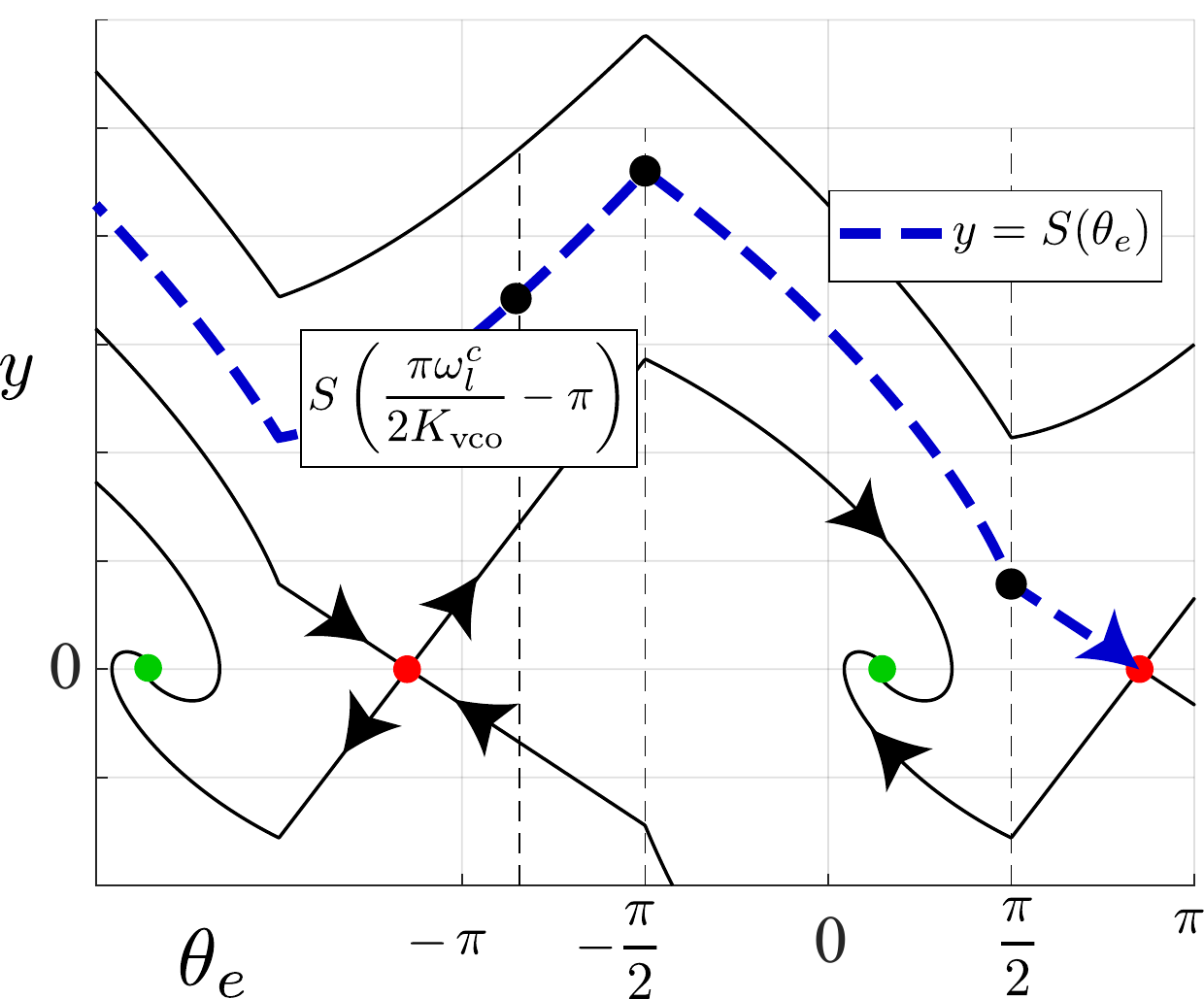}
	\caption{Separatrix $y = S(\theta_e)$ of the phase plane of \eqref{eq:PLL-triangular-after_change_of_variables} integration.
		Firstly, we compute $S(\frac{\pi}{2})$ and use it as the initial data of the Cauchy problem.
		Secondly, finding its solution in the domain B, we compute $S(-\frac{\pi}{2})$, which is used as the initial data of the Cauchy problem.
		Its solution in the domain A is used for the conservative lock-in frequency $\omega_l^c$ computation due to \eqref{eq:S and omega_l_c}.
		Parameters:
		$\tau_1 = 0.0633$,
		$\tau_2 = 0.0225$,
		$K_{\rm vco} = 250$,
		$\omega_e^{\rm free} = 73.732$.
	}
	\label{fig:separatrix integration}
\end{figure}

	The computation of $\omega_l^c$ consists of the following stages. 
	Let's divide the phase plane to the following domains:
	\begin{itemize}
		\item $A = \{(y,\ \theta_e) \mid \frac{\pi}{2} - 2\pi \le \theta_e(t) \le -\frac{\pi}{2}$; $\theta_e,y \in \mathbb{R}\}$,
		\item $B = \{(y,\ \theta_e) \mid -\frac{\pi}{2} \le \theta_e(t) \le \frac{\pi}{2}$; $\theta_e,y \in \mathbb{R}\}$.
	\end{itemize}
	In the open domains, system \eqref{eq:PLL-triangular-after_change_of_variables} is a linear one and can be integrated analytically.
	Firstly, we compute $S(\frac{\pi}{2})$, which is possible due to the continuity of \eqref{eq:PLL-model}, and use it as the initial data of the Cauchy problem (see Fig.~\ref{fig:separatrix integration}).
	Secondly, finding its solution in the domain B, we compute $S(-\frac{\pi}{2})$, which is used as the initial data of the Cauchy problem.
	Its solution in the domain A is used for the conservative lock-in frequency $\omega_l^c$ computation due to \eqref{eq:S and omega_l_c}.
	
	\subsection{S($\frac{\pi}{2})$ value}
		The saddle separatrix is locally described by the saddle's eigenvectors
		\begin{equation*}
		\begin{aligned}
		&
		V_+^{\rm s}
		=
		\begin{pmatrix}
		1\\
		-\kappa + \eta
		\end{pmatrix},
		\quad
		V_-^{\rm s}
		=
		\begin{pmatrix}
		1\\
		-\eta - \kappa
		\end{pmatrix}
		\end{aligned}
		\end{equation*}
		where
		\begin{equation*}
		\begin{aligned}
		&
		\eta
		=
		\frac{\frac{2}{\pi}\tau_2K_{\rm vco} - 1}{2\sqrt{\frac{2}{\pi}K_{\rm vco}(\tau_1+\tau_2)}},\\
		&\kappa = \sqrt{\eta^2 + 1}.
		\end{aligned}
		\end{equation*}

		Eigenvector $V^s_-$ points to a saddle and $V^s_+$ has the opposite direction.
		Since in the considered domain the system is a linear one, then the separatrix coincides with the line corresponding to $V_-^s$:
		\begin{equation}\label{eq:saddle_separatrix}
		\begin{aligned}
		&
		S(\theta_e) = (\kappa - \eta)(-\theta_e + \pi - \frac{\pi\omega_e^{\rm free}}{2K_{\rm vco}}),
		\quad \frac{\pi}{2} < \theta_e < \pi.
		\end{aligned}
		\end{equation}
		Let's obtain the limit value in $\theta_e = \frac{\pi}{2}$:
		\begin{equation*}
		\begin{aligned}
		& S(\frac{\pi}{2}) = (\kappa - \eta)(\frac{\pi}{2} - \frac{\pi\omega_e^{\rm free}}{2K_{\rm vco}}).
		\end{aligned}
		\end{equation*}

	\subsection{Analytical integration in domain B}
	
	In domain B, $v_e(\theta_e) = \frac{2}{\pi}\theta_e, v_e^\prime(\theta_e) = \frac{2}{\pi}$ and \eqref{eq:PLL-triangular-after_change_of_variables} can be rewritten as
	\begin{equation}\label{eq:linear equation in simmetric interval}
	\begin{aligned}
	&\dot y = - \theta_e - 2\xi y  + \frac{\pi\omega_e^{\rm free}}{2 K_{\rm vco}},\\
	&\dot \theta_e = y
	\end{aligned}
	\end{equation}	
	where
	\begin{equation*}
	\begin{aligned}
	&\xi 
	=
	\frac{\frac{2}{\pi}\tau_2 K_{\rm vco} + 1}{2\sqrt{\frac{2}{\pi}K_{\rm vco}(\tau_1+\tau_2)}} > 0.
	\end{aligned}
	\end{equation*}
	In the domains $\{y>0\}$ and $\{y<0\}$, variable $\theta_e(t)$ changes monotonically and the behaviour of system \eqref{eq:linear equation in simmetric interval} can be described by the first-order differential equation:
	\begin{equation}\label{eq:Chini equation in symmetric interval}
	\begin{aligned}
	& \frac{dy}{d\theta_e} = 
	-2\xi - \frac{\theta_e - \frac{\omega_e^{\rm free}}{\frac{2}{\pi} K_{\rm vco}}}{y}.
	\end{aligned}
	\end{equation}
	The obtained equation is Chini's equation \cite{Chini-1924, ChebK-2003}, which is a generalization of Abel and Riccati equations.
	The change of variables $z = \frac{y}{\theta_e - \frac{\omega_e^{\rm free}}{\frac{2}{\pi} K_{\rm vco}}}$ maps equation \eqref{eq:Chini equation in symmetric interval} into a separable one:
	\begin{equation}\label{eq:PLL-equation-to_integrate in symmetric interval}
	\begin{aligned}
	& \frac{zdz}{z^2 + 2\xi z + 1} = -\frac{d\theta_e}{\theta_e - \frac{\omega_e^{\rm free}}{\frac{2}{\pi} K_{\rm vco}}}.
	\end{aligned}
	\end{equation}
	If $\theta_e\ne \frac{\omega_e^{\rm free}}{\frac{2}{\pi} K_{\rm vco}}$ and $z^2 + 2\xi z + 1\ne0$ then the solutions of system \eqref{eq:Chini equation in symmetric interval} and system \eqref{eq:PLL-equation-to_integrate in symmetric interval} coincide in domains $0 < \theta_e < \frac{\pi}{2}$ and $-\frac{\pi}{2} < \theta_e < 0$.
	Depending on the type of an asymptotically stable equilibrium, the following cases appear:
		\begin{itemize}
		\item $\xi > 1$ (the equation $z^2 + 2\xi z + 1 = 0$ corresponds the eigenvectors of the stable node),
		\item $\xi = 1$ (the equation $z^2 + 2\xi z + 1 = (z + \xi)^2 = 0$ corresponds the eigenvector of the stable degenerate node),
		\item $0 < \xi < 1$ (here the case $z^2 + 2\xi z + 1 = 0$ is not possible).
	\end{itemize}
	
	It can be shown that if $\xi\ge1$ then in domain $B$ separatrix $y = S(\theta_e)$ satisfies $N(y, \theta_e) = N(S(\frac{\pi}{2}),\ \frac{\pi}{2})$ where
	\begin{equation}\label{eq:N solution case 1 and 2}
	\begin{aligned}
	& N(y, \theta_e) 
	=
	\frac{1}{2}
	\ln
	\Big(
	\big(y + (\xi - \rho)(\theta_e - \frac{\pi\omega_e^{\rm free}}{2 K_{\rm vco}})\big)^{\frac{\rho - \xi}{\rho}}
	\cdot\\
	&\cdot
	\big(y + (\xi + \rho)(\theta_e - \frac{\pi\omega_e^{\rm free}}{2 K_{\rm vco}})\big)^{\frac{\rho + \xi}{\rho}}
	\Big), \quad \xi>1,\\
	&N(y, \theta_e) 
	=
	\frac{\theta_e - \frac{\pi\omega_e^{\rm free}}{2 K_{\rm vco}}}{y + \theta_e - \frac{\pi\omega_e^{\rm free}}{2 K_{\rm vco}}} + \ln(2|y + \theta_e - \frac{\pi\omega_e^{\rm free}}{2 K_{\rm vco}}|),
	\quad \xi = 1,\\
	&\rho = \sqrt{|\xi^2 - 1|}.
	\end{aligned}
	\end{equation}
	Similarly, if $\xi < 1$ then
	\begin{itemize}
		\item in domain $-\frac{\pi}{2} < \theta_e < \frac{\pi\omega_e^{\rm free}}{2K_{\rm vco}}$ separatrix $y = S(\theta_e)$ satisfies $N(y, \theta_e) = N(S(\frac{\pi}{2}),\ \frac{\pi}{2})  + \frac{\pi \xi}{\rho}$,
		\item in domain $\frac{\pi\omega_e^{\rm free}}{2K_{\rm vco}} < \theta_e < \frac{\pi}{2}$ separatrix $y = S(\theta_e)$ satisfies $N(y, \theta_e) = N(S(\frac{\pi}{2}),\ \frac{\pi}{2})$
	\end{itemize} 
where
	\begin{equation}\label{eq:N solution case 3}
\begin{aligned}
& N(y, \theta_e)
=
\frac{1}{2}\ln(y^2 + 2\xi y(\theta_e - \frac{\pi\omega_e^{\rm free}}{2K_{\rm vco}}) + (\theta_e - \frac{\pi\omega_e^{\rm free}}{2K_{\rm vco}})^2)
- \\
&-
\frac{\xi}{\rho}
\arctan\Big(
\frac{y + \xi(\theta_e - \frac{\pi\omega_e^{\rm free}}{2K_{\rm vco}})}
{(\theta_e - \frac{\pi\omega_e^{\rm free}}{2K_{\rm vco}})\rho}
\Big).
\end{aligned}
\end{equation}
	
Let us denote $y_{\rm AB} = S(-\frac{\pi}{2})$ and use this value as the initial data of the Cauchy problem:
\begin{equation}\label{eq:N equations for B domain}
\begin{aligned}
&
N\Big(y_{\rm AB},\ -\frac{\pi}{2}\Big) 
=
N(S(\frac{\pi}{2}),\ \frac{\pi}{2}) \quad \xi \ge 1,\\
&
N\Big(y_{\rm AB},\ -\frac{\pi}{2}\Big) 
=
N(S(\frac{\pi}{2}),\ \frac{\pi}{2}) + \frac{\pi \xi}{\rho} \quad \xi < 1.
\end{aligned}
\end{equation}
	Taking into account equations \eqref{eq:N solution case 1 and 2} and \eqref{eq:N solution case 3}, equations \eqref{eq:N equations for B domain} provide the last three formulae in \eqref{eq:concervative lock-in}.

	\subsection{Analytical integration in domain A}	
	
	In domain A, $v_e(\theta_e) = -\frac{2}{\pi}\theta_e + 2, v_e^\prime(\theta_e) = -\frac{2}{\pi}$ and \eqref{eq:PLL-triangular-after_change_of_variables} can be rewritten as
	\begin{equation}\label{eq:linear equation in nonsimmetric interval}
	\begin{aligned}
	&\dot y = (\theta_e + \pi) + 2\eta y  + \frac{\pi\omega_e^{\rm free}}{2K_{\rm vco}},\\
	&\dot \theta_e = y.
	\end{aligned}
	\end{equation}	
	In the domains $\{y>0\}$ and $\{y<0\}$, variable $\theta_e(t)$ changes monotonically and the behaviour of system \eqref{eq:linear equation in nonsimmetric interval} can be described by the first-order differential equation:
	\begin{equation}\label{eq:Chini equation in nonsymmetric interval}
	\begin{aligned}
	& \frac{dy}{d\theta_e} = 
	\frac{2}{\mu}\eta + \frac{\theta_e + \pi + \frac{\pi\omega_e^{\rm free}}{2K_{\rm vco}}}{ y}.
	\end{aligned}
	\end{equation}
	The change of variables $z = \frac{y}{\theta_e + \pi +  \frac{\pi\omega_e^{\rm free}}{2K_{\rm vco}}}$ maps equation \eqref{eq:Chini equation in nonsymmetric interval} into a separable one:
	\begin{equation}\label{eq:PLL-equation-to_integrate in nonsymmetric interval}
	\begin{aligned}
	& \frac{z dz}{z^2 - 2\eta z - 1} = -\frac{d\theta_e}{\theta_e + \pi + \frac{\pi\omega_e^{\rm free}}{2K_{\rm vco}}}.
	\end{aligned}
	\end{equation}
	If $\theta_e\ne -  (\frac{\pi\omega_e^{\rm free}}{2K_{\rm vco}} + \pi)$ and $z^2 - 2\eta z - \ne 0$ then the solutions of system \eqref{eq:Chini equation in nonsymmetric interval} and system \eqref{eq:PLL-equation-to_integrate in nonsymmetric interval} coincide.
	
	It can be shown that in domain $A$ separatrix $y = S(\theta_e)$ satisfies $M(y, \theta_e) = M(y_{\rm AB},\ -\frac{\pi}{2})$ where
	\begin{equation}\label{eq:M}
	\begin{aligned}
	&M(y, \theta_e) 
	=
	\frac{1}{2}
	\ln 
	\Big(
	\Big(
	y + \frac{\theta_e + \pi +  \frac{\pi\omega_e^{\rm free}}{2K_{\rm vco}}}{\kappa + \eta}
	\Big)^{\frac{\kappa - \eta}{\kappa}}
	\cdot\\
	&\cdot
	\Big(
	y +  \frac{\theta_e + \pi +  \frac{\pi\omega_e^{\rm free}}{2K_{\rm vco}}}{\eta - \kappa}
	\Big)^{\frac{\kappa + \eta}{\kappa}}
	\Big).
	\end{aligned}
	\end{equation}
	Finally, the first equation in \eqref{eq:concervative lock-in} is obtained by consideration \eqref{eq:S and omega_l_c} and \eqref{eq:M}:
	\begin{equation*}
	\begin{aligned}
	&
	M\left(2\omega_l^c\sqrt{\frac{\tau_1 + \tau_2}{k K_{\rm vco}}},\ \frac{(\pi - \frac{1}{k})\omega_l^c}{K_{\rm vco}} - \pi\right)
	= 
	M(y_{\rm AB},\ -\frac{1}{k}).
	\end{aligned}
	\end{equation*}
	
\end{proof}

\section*{Appendix B: Calculation of the conservative lock-in frequency}
Calculation of the conservative lock-in frequency for PLL with lead-lag filter and triangular phase-detector characteristic.
\begin{lstlisting}
function out = omega_l_conservative(tau_1, tau_2, k, K_vco)
    
    out = 0;
    mu = pi*k - 1;
    xi = (k*tau_2*K_vco + 1)/(2*sqrt(k*K_vco*(tau_1 + tau_2)));
    eta = (k*tau_2*K_vco - mu)/(2*sqrt(k*K_vco*(tau_1 + tau_2)));
    rho = sqrt(abs(xi^2 - 1));
    kappa = sqrt(eta^2 + mu);

    syms y_ab zomega_lc;

    curve1 = (2*zomega_lc)^2*...
        (sqrt((tau_1 + tau_2)/(k*K_vco)) - ...
        (eta - kappa)/(k*K_vco))^((kappa - eta)/kappa)*...
        (sqrt((tau_1 + tau_2)/(k*K_vco)) - ...
        (eta + kappa)/(k*K_vco))^((kappa + eta)/kappa)  == ...
        (y_ab - (eta - kappa)*((zomega_lc + K_vco)/...
        (k*K_vco)))^((kappa - eta)/kappa)*...
        (y_ab - (eta + kappa)*((zomega_lc + K_vco)/...
        (k*K_vco)))^((kappa + eta)/kappa);

    if xi > 1
        curve2 = (y_ab - (xi - rho)*((zomega_lc + K_vco)/...
        (k*K_vco)))^((rho - xi)/(rho))*...
            (y_ab - (xi + rho)*...
            ((zomega_lc + K_vco)/(k*K_vco)))^((rho + xi)/(rho)) == ...
            (kappa - eta + xi - rho)^((rho - xi)/...
            (rho))*(kappa - eta + xi + rho)^((rho + xi)/(rho))*...
            ((K_vco - zomega_lc)/(k*K_vco))^2;
    else 
        if (abs(xi - 1) < 0.001)
            curve2 = (-(K_vco + zomega_lc)/(k*K_vco))/...
            (y_ab -(K_vco + zomega_lc)/(k*K_vco)) + ...
                log(2*abs(y_ab -(K_vco + zomega_lc)/(k*K_vco))) == ...
                1/(kappa - eta + 1) + ...
                ln(2*(kappa - eta + 1)*(K_vco - zomega_lc)/(k*K_vco));
        else
            curve2 = 1/2*log(y_ab^2 - 2*xi*y_ab*(K_vco + zomega_lc)/...
            (k*K_vco) + ((K_vco + zomega_lc)/(k*K_vco))^2) -...
                xi/rho*atan((y_ab - xi*(K_vco + zomega_lc)/(k*K_vco))/...
                (-(K_vco + zomega_lc)/(k*K_vco)*rho)) == ...
            1/2*log(...
            ((kappa - eta)^2 + 2*xi*(kappa-eta) + 1)*...
            ((K_vco - zomega_lc)/(k*K_vco))^2 ...
            ) - xi/rho*atan((kappa - eta + xi)/rho) + pi*xi/rho;
        end
    end

    res = vpasolve([curve1, curve2], [0 Inf; 0 K_vco]);
    if ~isempty(eval(res.zomega_lc))
        out = eval(res.zomega_lc);
    end
end

\end{lstlisting}



\end{document}